\documentclass{amsart}
\usepackage{amssymb}
\usepackage{latexsym}
\usepackage{cmll}
\usepackage{stmaryrd}
\usepackage{eqnarray}
\usepackage{mathtools}
\usepackage{xspace}
\usepackage{algorithm}
\usepackage{algpseudocode}

\usepackage{multirow}
\usepackage{lineno}

\newtheorem{lm}{Lemma}[section]
\newtheorem{thm}[lm]{Theorem}

\newtheorem{df}[lm]{Definition}

\theoremstyle{remark}
\newtheorem{claim}{Claim}

\newcommand{\m}[1]{{\uppercase {\mathbf{#1}}}}
\newcommand{\rel}[1]{{\uppercase {\mathbb{#1}}}}

\DeclareMathOperator{\num}{num}
\DeclareMathOperator{\dep}{dep}

\newcommand{\ceqv}[1]{\ensuremath{\operatorname{\textsc{Ceqv}
                                \ifthenelse{\equal{#1}{}}{}{\!\left( {\m #1} \right)}}}}
\newcommand{\ceqvt}[1]{\ensuremath{\operatorname{\textsc{Ceqv}{}_{T}
                                \ifthenelse{\equal{#1}{}}{}{\!\left( {\m #1} \right)}}}}
\newcommand{\ceqvm}[1]{\ensuremath{\operatorname{\textsc{Ceqv}{}_{TM}
                                \ifthenelse{\equal{#1}{}}{}{\!\left( {\m #1} \right)}}}}
\newcommand{\csat}[1]{\ensuremath{\operatorname{\textsc{Csat}
                                \ifthenelse{\equal{#1}{}}{}{\!\left( {\m #1} \right)}}}}
\newcommand{\csatt}[1]{\ensuremath{\operatorname{\textsc{Csat}{}_{T}
                                \ifthenelse{\equal{#1}{}}{}{\!\left( {\m #1} \right)}}}}
\newcommand{\csatm}[1]{\ensuremath{\operatorname{\textsc{Csat}{}_{TM}
                                \ifthenelse{\equal{#1}{}}{}{\!\left( {\m #1} \right)}}}}
\newcommand{\Csat}[1]{\ensuremath{\operatorname{\textsc{Csat}
                                \ifthenelse{\equal{#1}{}}{}{\!\left( { #1} \right)}}}}
\newcommand{\mcsat}[1]{\ensuremath{\operatorname{\textsc{MCsat}
                                \ifthenelse{\equal{#1}{}}{}{\!\left( {\m #1} \right)}}}}
\newcommand{\scsat}[1]{\ensuremath{\operatorname{\textsc{SCsat}
                                \ifthenelse{\equal{#1}{}}{}{\!\left( {\m #1} \right)}}}}
\newcommand{\csp}[1]{\ensuremath{\operatorname{\textsc{CSP}
                                \ifthenelse{\equal{#1}{}}{}{\!\left( {\rel #1} \right)}}}}

\newcommand{\polsat}[1]{\ensuremath{\operatorname{\textsc{PolSat}
                                \ifthenelse{\equal{#1}{}}{}{\!\left( {\m #1} \right)}}}}
\newcommand{\poleq}[1]{\ensuremath{\operatorname{\textsc{PolEqv}
                                \ifthenelse{\equal{#1}{}}{}{\!\left( {\m #1} \right)}}}}

\newcommand{\npc}{\textsf{NP}-complete\xspace}
\newcommand{\conpc}{\textsf{co-NP}-complete\xspace}

\newcommand{\ptime}{\textsf{P}\xspace}

\newcommand{\Z}{\mathbb{Z}}

\newcommand{\pro}[2]{\ensuremath{\Pi_{#1}
                                \ifthenelse{\equal{#1}{}}{}{\!\left( { #2} \right)}}}

\newcommand{\set}[1]{{\left\{ {#1} \right\} }}

\renewcommand{\leq}{\leqslant}
\renewcommand{\geq}{\geqslant}
\renewcommand{\le}[1]{\leqslant_{#1}}
\renewcommand{\ge}[1]{\geqslant_{#1}}

\newcommand{\Usum}{\sum}
\newcommand{\Lsum}{\sum}
\newcommand{\Uplus}{+}
\newcommand{\Uminus}{-}
\newcommand{\Lplus}{+}

\newcommand{\pol}[1]{{\rm Pol\:\m #1}}
\newcommand{\Pol}[1]{{\rm Pol( #1)}}
\newcommand{\poln}[2]{{\rm Pol}_{#1}\m #2}

\newcommand{\po}[1]{{\mathbf {#1}}}

\renewcommand{\o}[1]{\overline {#1}}
\newcounter{ttable}

\newcommand{\map}{\longrightarrow}

\newcommand{\h}[1]{\widehat{#1}}


\begin{document}
\title{Circuit equivalence in 2-nilpotent algebras}

\author{Piotr Kawa\l{}ek, Michael Kompatscher, Jacek Krzaczkowski}

\address{Piotr Kawa\l{}ek\\ Jagiellonian University\\  Faculty of Mathematics and Computer Science\\Department of Theoretical Computer Science\\
ul. Prof. S. \L{}ojasiewicza 6\\ 30-348, Krak\'ow, Poland}
\email{piotr.kawalek@doctoral.uj.edu.pl}

\address{Michael Kompatscher, Charles University Prague\\ MFF\\ Department of Algebra\\Sokolovka 83\\ 186 75 Praha 8\\ Czech Republic}
\email{ michael@logic.at}

\address{Jacek Krzaczkowski\\ Maria Curie-Sk\l{}odowska University\\ Faculty of Mathematics, Physics and Computer Science\\Department of  Computer Science\\
ul.Akademicka 9\\ 20-033, Lublin, Poland}
\address{Jagiellonian University\\  Faculty of Mathematics and Computer Science\\Department of Theoretical Computer Science\\
ul. Prof. S. \L{}ojasiewicza 6\\ 30-348, Krak\'ow, Poland}
\email{krzacz@poczta.umcs.lublin.pl}

\thanks{The first and the third authors are partially supported by Polish NCN Grant \#~2014/14/A/ST6/00138\\
The second author is supported by Charles University Research Centre programs PRIMUS/SCI/12 and UNCE/SCI/022 as well as grant 18-20123S of the Czech Grant Agency (GA\v{C}R)}

\maketitle

\begin{abstract}
The circuit equivalence problem of a finite algebra $\m A$ is the computational problem of deciding whether two circuits over $\m A$ define the same function or not. This problem not just generalises the equivalence problem for Boolean circuits, but is also of high interest in universal algebra, as it models the problems of checking identities in $\m A$. In this paper we discuss the complexity for algebras from congruence modular varieties. A partial classification was already given in \cite{IK}, leaving essentially only a gap for nilpotent but not supernilpotent algebras. We start a systematic study of this open case, proving that the circuit equivalence problem is in \ptime for $2$-nilpotent such algebras.
\end{abstract}

\section{Introduction}

To solve equations is one of the oldest and best-known problems in mathematics. For many centuries it inspired research in algebra and lead both to the development of new theoretical concepts and new algorithms (let us only mention Galois theory, Diophantine Equations and Gaussian elimination). From a computer science point of view the main focus lies to the latter and the question: What is the computational complexity of solving equations in a given algebra $\m A$?

More formally, by the equation satisfiability problem \polsat{A} of a fixed algebra $\m A$ we denote the computational problem of deciding whether a given equation of polynomials over $\m A$ has a solution or not. A prominent example of such a problem is \polsat{{}\mathbb Z,+,\cdot}, the problem of deciding whether a Diophantine equation has a solution, which was proven to be undecidable by Matiyasevich \cite{matiyasevich:10th}. 

The equivalence problem \poleq{A} is the closely related problem, where the input consists of two polynomials over $\m A$, and the task is to decide whether they define the same function. In other words the task is to check if an equation holds for \emph{all} possible assignments of values to the variables. For finite algebras \polsat{} clearly is in $\textsf{NP}$ and \poleq{} in $\textsf{co}$-$\textsf{NP}$; in the last twenty years there were numerous papers further investigating the complexity and trying to find hardness and tractability criteria for both problems (e.g. \cite{aichinger-systems}, \cite{aichinger-mudrinski}, \cite{burris-lawrence:groups}, \cite{goldmann-russell}, \cite{Horvath-EqSolvNilp}, \cite{horvath-szabo:groups}, \cite{horvath-szabo:polsatstar}, \cite{klima:monoids}, \cite{schwarz}).

One of the major obstacles in studying \polsat{A} and \poleq{A} systematically for all finite algebras is that the complexity strongly depends on the signature of $\m A$. For example, $A_4$ and some other solvable, non-nilpotent groups are known to induce problems \polsat{} and \poleq{} that are in P; however after adding the commutator $[x,y]=x^{-1}y^{-1}xy$ as a basic operation we obtain \npc \polsat{} problems and \conpc \poleq{} problems \cite{horvath-szabo:a4} \cite{komp-group}. Roughly speaking this results from the fact that some operations can be written in a much more concise ways using commutators than just the group operations alone. In fact, the terms used in proving NP-completeness inflate to exponentially longer expressions in the pure group language.

To resolve this problem, it was recently proposed to encode an input equation by circuits \cite{IK}. This approach prevents an artificial inflation of the input as in the above example. Consequently the complexity for these `circuit problems' only depends on the set of polynomial operations of the algebra, allowing for the use of universal algebra in studying their complexity. 
We formally define the circuit satisfiability (\csat{}) and circuit equivalence (\ceqv{}) as follows:
\begin{itemize}
 \item \csat{A}\\
 given a circuit over the algebra $\m A$ with two output gates $g_1$, $g_2$ is there a valuation of input gates $\overline{x}= (x_1,\ldots,x_n)$ that gives the same output on both $g_1$ and $g_2$, i.e. $g_1(\overline{x}) =g_2(\overline{x})$?
 \item \ceqv{A}\\
 given a circuit over the algebra $\m A$ is it true that for all inputs $\overline{x}$ we have the same values on given two output gates $g_1$, $g_2$, i.e. $g_1(\overline{x}) =g_2(\overline{x})$ for all $\overline{x}\in A^n$?
\end{itemize}
Besides \cite{IK} these problems were also considered in \cite{IKK} and \cite{aichinger-systems} (and implicitly already earlier, e.g. in \cite{horvath-szabo:polsatstar}). In \cite{IK} Idziak and the third author set the goal to classify the computational complexity of \csat{} and \ceqv{} for algebras from congruence modular varieties. On one hand these algebras form a quite broad class with many elements of interest in classical algebra such as groups, quasigroups, rings, modules, fields, lattices, Boolean algebras. On the other hand there is well-developed theory of commutators in this case, which will be the basis of our proof.

There are strong indications that the complexity hierarchy of \ceqv{} in the congruence modular case corresponds to a structural hierarchy in commutator theory: By \cite{IK}, for every non-nilpotent algebra $\m A$ from a congruence modular variety there exists a quotient algebra $\m A'$ of $\m A$ such that \ceqv{A'} is \conpc. On the other hand it was shown in \cite{aichinger-mudrinski} that \ceqv{} for so called supernilpotent algebras from  congruence modular varieties is in \ptime. 

We remark that in congruence modular varieties supernilpotent algebras are strictly contained in nilpotent algebras (but it is not true in general, see \cite{moore-moorhead}). This leaves a gap for nilpotent, but not supernilpotent algebras. In \cite{IKK} an example of a $2$-nilpotent, but not supernilpotent algebra $\m A$ was given for which \ceqv{A} can be solved in polynomial time.

This paper is the first step in the systematic study of \ceqv{} for all nilpotent algebras. We prove that \ceqv{A}  is in \ptime for every $2$-nilpotent algebra $\m A$ from a congruence modular variety. Our algorithm is based on the analysis of a normal form of polynomial operations of such algebras. Thus it comes hand in hand with a deeper understanding of the structure of $2$-nilpotent algebras. Our hope is to generalise these results to $k$-nilpotent algebras in future research.

\section{Definitions and notation}
We are going to use standard notation from universal algebra, which can be found in \cite{fm}. We define a signature $F$ to be a sequence $(f_i, k_i)_{i \in I}$, where each $f_i$ is a function symbol and $k_i$ is the arity corresponding to this symbol. An algebra over signature $F$ is then a tuple $\m A= (A, (f^{\m A}_i)_{i \in I})$ for some set $A$ and $f^{\m A}_i$ being a function from  $A^{k_i}$ to $A$. Each $f^{\m A}_i$ will be called a \textit{basic operation} of $\m A$. A \emph{finite algebra} is an algebra with finite universe $A$ and finite signature, so it has finitely many basic operations. An algebra $\m B$ is a \emph{subalgebra} of $\m A$ iff $B \subseteq A$, $B$ is closed under all basic operations of $\m A$, and the basic operations of $\m B$ are the~basic operations of $\m A$ restricted to the set $B$. In this case we write $\m B \leq \m A$.

For an algebra $\m A$, let us denote by the \emph{clone of polynomials} $\pol{A}$ the smallest set of operations on $A$ that contains all constant functions, all projections $\pi_{i}^n(x_1, \ldots, x_n) = x_i$, all basic operations of $\m A$ and that is closed under composition. Moreover let $\poln{n}{A}$ be the set of $n$-ary functions in $\pol{A}$. It is straightforward to see that for finite algebras $\ceqv{A}$ reduces to $\ceqv{B}$ if $\pol{A} \subseteq \pol{B}$ (see \cite{IK}). We will say that $\m B$ and $\m A$ are \emph{polynomially equivalent} iff there exist an algebra $\m B'$ isomorphic to $\m B$ with $\pol{A} = \pol{B'}$.

An \emph{affine} algebra is an algebra that is polynomially equivalent to a module. A \emph{Maltsev operation} is ternary operation $d(x,y,z)$ such that $d(x, x, y) = y$ and $d(x, y, y) = x$ holds for all $x,y$. For instance, every affine algebra has $x-y+z$ as a Maltsev operation.

We will use lowercase overlined letter $\overline{x}$ to denote tuples $\overline{x}=(x_1,\ldots,x_n) \in U^n$. In our paper $U$ will often stand for a direct product $\mathbb{Z}_{p_1^{k_1}}\times\ldots\times\mathbb{Z}_{p_m^{k_m}}$. In this case, for every $i = 1,\ldots,n$ we further use the notation $x_i=(x_i^{(1)},\ldots,x_i^{(m)})$, with $x_i^{(j)} \in \mathbb{Z}_{p_j^{k_j}}$. In particular, if we just want to study the $\mathbb{Z}_{p_j^{k_j}}$-component of a tuple $\overline{x} \in U^n$ we will use the notion $x^{(j)}=(x_1^{(j)},\ldots,x_n^{(j)})$.

\section{The structure of 2-nilpotent algebras}
In this section we provide some structural background on 2-nilpotent algebras and prove that (in some of them) we can represent polynomials in a certain normal form.

Nilpotent algebras can be defined using the commutator of congruences, generalising the notion of nilpotent groups and rings. We are however not going to give the original definition here and refer to the book \cite{fm} for background. For our purposes it will be enough to give a characterisation of nilpotent algebras in congruence modular varieties. In this case commutator theory works especially well and allows us to obtain much structural information about algebras. It is for instance well known that Abelian (or '1-nilpotent') algebras exactly correspond to affine algebras.

Now 2-nilpotent algebras from congruence modular varieties can be considered as the action of one affine algebra on an other one (see Chapter VII of \cite{fm}). More precisely, for two algebras $\m U$ and $\m L$ of the same signature $F$ such that
\begin{itemize}
 \item $\m U$ is polynomially equivalent to a module $(U;\Uplus)$ over a ring $\m R_U$, and
 \item $\m L$ is polynomially equivalent to a module $(L;\Lplus)$ over a ring $\m R_L$,
\end{itemize}
and a set $\widehat{F}$ of functions such that for every $f\in F$, say $k$-ary, there is $\widehat{F}\ni\widehat{f}\colon U^k\map L$ we define $\m L \otimes^{\widehat{F}} \m U$ as an algebra over signature $F$ and universe $L\times U$ by
\begin{equation}
\label{eq-flu}
f^{\m L \otimes^{\widehat{F}}\m U}((l_1,u_1),\ldots,(l_k,u_k))
=
(f^{\m L}(l_1,\ldots,l_k)+\h f(u_1,\ldots,u_k),f^{\m U}(u_1,\ldots,u_k)).
\end{equation}
It is shown in \cite{fm} that every 2-nilpotent algebra over signature $F$ from a congruence-modular variety is isomorphic to some $\m L \otimes^{\widehat{F}} \m U$. Working in such $\m L \otimes^{\widehat{F}} \m U$ we are going to show that every polynomial (or circuit) of it can be expressed in a certain normal form which will be extensively used by our polynomial time algorithm.

First of all observe that not only basic operations of $\m L \otimes^{\widehat{F}} \m U$, but all its polynomial operations, can be expressed in form \eqref{eq-flu}. Moreover, since $\m L$ and $\m U$ are affine, for a polynomial operation $\po p$ over $\m L \otimes^{\widehat{F}} \m U$ there exist $\lambda_i$, $\alpha_i$, $u_0$ such that

\[
\po p^{\m L \otimes^{\widehat{F}}\m U}((l_1,u_1),\ldots,(l_k,u_k))
=
\left(\Lsum_{i=1}^k \lambda_i l_i \Lplus \widehat{\po p}(u_1,\ldots,u_k),
\Usum_{i=1}^k \alpha_i u_i+u_0 \right).
\]
Let $(U,\Uplus) \cong \mathbb{Z}_{p_1^{k_1}}\times\dots\times\mathbb{Z}_{p_m^{k_m}}$ be the underlying group of $\m U$. We will prove that if $\m U$ and $\m L$ are of coprime order then $\widehat{\po p}(u_1,\ldots,u_k)$ can be presented as a sum of expressions in the form 
\[
\mu\cdot \po w^a_{1,\ldots,1}\left(\Usum_{i=1}^k\beta^{(1)}_i u_i \Uplus u^{(1)}_0 ,\ldots,
              \Usum_{i=1}^k\beta^{(s)}_i u_i \Uplus u^{(s)}_0\right),
\]
where  $\po w_{n_1,\ldots,n_m}^a(\overline{x})$ is a function from
$(\Z_{p_1^{k_1}})^{n_1}\times \ldots \times (\Z_{p_m^{k_m}})^{n_m} \rightarrow L$, $a \in L$ and
$$
\po  w^a_{n_1,\ldots,n_m}(\overline{x}) = \left\{ \begin{array}{ll}
a & \textrm{if $x_j^{(i)} = 0$ for all $i= 1,\ldots,m$ and $j=1,\ldots,n_i$},\\
0 & \textrm{otherwise}.\\
\end{array} \right.
$$

For short we will write $\po w^a$ for $\po w_{1,\ldots,1}^a$. Notice that $\po w^a$ is a function from $U \rightarrow L$ and $\po w_{n,\ldots,n}^a(\overline{x})$ can be interpreted as operation $U^n \rightarrow L$.
For  $\beta,\o x \in (\mathbb{Z}_{p_1^{k_1}} \times \mathbb{Z}_{p_2^{k_2}} \times \ldots \times \mathbb{Z}_{p_m^{k_m}})^n$ we will use the following notation:
\[
\beta \odot \overline{x} = (\Usum_{i=1}^n \beta_i^{(1)} x_i^{(1)}, \Usum_{i=1}^n \beta_i^{(2)} x_i^{(2)}, \ldots, \Usum_{i=1}^n \beta_i^{(m)} x_i^{(m)}).
\]

If $|L|$ and $|U|$ are co-prime we can express $\widehat{\po p}$ in a normal form, just using $\po w^a$:

\begin{lm}
\label{normal-lemma}
Let $\m U$, $\m L$ be modules such that $(U,\Uplus)$ is isomorphic to $\mathbb{Z}_{p_1^{k_1}} \times \mathbb{Z}_{p_2^{k_2}} \times \ldots \times \mathbb{Z}_{p_m^{k_m}}$ and $|U|$ and $|L|$ are coprime. Then every function $f\colon U^n\map L$ can be expressed in the form:
\begin{equation}\label{eq-normal}
f(x_1,\ldots,x_n)=\Lsum_{l\in L, c\in U\atop \beta\in U^n}\mu^{l}_{\beta,c} \po w^{l}(\beta \odot \o x \Uplus c).
\end{equation}
\end{lm}
\begin{proof}
For $l\in L$ we set $s_{l}(x_1,\ldots,x_n) = \po w^l_{n,\ldots,n}(x_1^{(1)},\ldots,x_n^{(1)},\ldots,x_1^{(m)},\ldots,x_n^{(m)})$.
Observe that $s_{l}(x_1,\ldots,x_n)=l$ if $x_1=\ldots=x_n=0$ and $s_{l}(x_1,\ldots,x_n)=0$ otherwise. Then clearly every function $f(x_1,\ldots,x_n)$ can be written as the sum of all expressions $s_{f(u_1,\ldots,u_n)}(x_1\Uminus u_1,\ldots,x_n\Uminus u_n)$, for all $(u_1,\ldots,u_n) \in U^n$. Hence to prove the statement of the lemma it suffices to show that for all indices $i_1,\ldots,i_k$ and all $l \in L$ we are able to express $\po w_{i_1,\ldots,i_k}^l$ in the form \eqref{eq-normal}. 

First, we will prove this for the case $m=1$. For convenience we will write $p = p_1$ and $k = k_1$. If $k=1$ then we can obtain $\po w_{n+1}$ using $\po w_{n}$ in the following way:
\begin{equation} \label{eq-normal-k=1}
\begin{split}
\po w^l_{n+1}(x_1,\ldots,x_{n+1})=\nu_{p,l}(\Lsum_{i=0}^{p-1}\po w^l_n(x_1,\ldots,x_{n-1},x_{n}\Uplus ix_{n+1}) \\
-\Lsum_{i=1}^{p-1}\po w^l_n(x_1,\ldots,x_{n-1},i\Uplus x_{n+1})),
\end{split}
\end{equation}
where $\nu_{p,l}$ is a scalar from $\m R_{\m L}$ inverse to $p$ (i.e. scalar equivalent to an endomorphism $e_{p}$ of $(L,\Lplus)$ such that $e_{p}(x)=\underbrace{x\Lplus\ldots\Lplus x}_p$). We can assume that such inverse scalar exists since $p=|U|$ is coprime to $|L|$ (in fact, we can assume that $R_{L}$ contains all endomorphisms of $\m L$). A straightforward computation shows that the identity \eqref{eq-normal-k=1} indeed holds (see also Lemma 3.1. in \cite{IKK}).

For arbitrary $k$ we prove the statement by induction. So let us assume it holds for all $k'<k$. Again, it is enough to show that we can write $\po w_{n+1}^l$ in the form \eqref{eq-normal}. As an intermediate step, let us define the polynomial $\po t^l_{n+1}(x_1,\ldots,x_{n+1})$ by the sum
\begin{equation} \label{eq-tsum}
\Lsum_{i=0}^{p^{k}-1}\po w^l_n(x_1,\ldots,x_{n-1},x_{n}\Uplus ix_{n+1})
+\Lsum_{i=0}^{p^{k-1}-1}\po w^l_n(x_1,\ldots,x_{n-1},pix_{n}\Uplus x_{n+1}).
\end{equation}
If there is an index $j<n$ such that $x_j\not=0$ then $\po t^l_{n+1}$ is equal $0$. 
We next give a description of $\po t^l_{n+1}$ in the remaining case $x_1 = \ldots = x_{n-1}= 0$. Let $o(x_n)$ be the order of $x_n$ in the group theoretical sense. Notice that the first sum in \eqref{eq-tsum} counts the number of indices $i=0,\ldots,p^{k}-1$ such that $x_n + ix_{n+1}$ is $0$. This value is $0$ if $o(x_{n+1}) < o(x_{n})$ and $p^{k}/o(x_{n+1})$ otherwise. The second sum in \eqref{eq-tsum} counts the number of indices $i= 0,\ldots,p^{k-1}-1$ such that $pix_n + x_{n+1}$ is $0$. It is also easy to see that this value is $0$ if $o(x_{n}) \leq o(x_{n+1}) \neq 1$ and $p^{k-1}/o(x_{n})$ otherwise.

The above analysis shows in particular that if $x_n\not=0$ or $x_{n+1}\not=0$ the value of $\po t_{n+1}^l(0,\ldots,0,x_n,x_{n+1})$ only depends on the values of $x_{n}$ and $x_{n+1}$ modulo $p^{k-1}$. Moreover, if $p^{k-1}$ divides $x_n$ and $x_{n+1}$, then $t_{n+1}^l(0,\ldots,0,x_n,x_{n+1})$ is equal to $(p^k+p^{k-1})l$ if $x_n = x_{n+1} = 0$ and $p^{k-1} l$ else. Hence, 
\[
p^{k-1} \po w^l_{m+1}(x_1,\ldots,x_{n+1})= \po t^l_{m+1}(x_1,\ldots,x_{n+1})+ \po r^l_{n+1}(x_1,\ldots,x_{n+1}) 
\]
where $\po r_{n+1}^l(x_1,\ldots, x_{n+1}) = 0$ if there is a $j < n$ with $x_j \neq 0$, and $\po r_{n+1}^l(0,\ldots,0,x_n, x_{n+1})$ is a function that only depends on the value of $x_n$ and $x_{n+1}$ modulo $p^{k-1}$. In other words $\po r_{n+1}^l(0,\ldots,0,x_n, x_{n+1})$ can be seen as an operation from the submodule $pU$ to $L$. As the group structure of $pU$ is $\mathbb{Z}_{p^{k-1}}$, by induction hypothesis we can express 
$\po r_{n+1}^l$ using a normal form as in \eqref{eq-tsum}. This and the observation that $p^{k-1}$ has an inverse in $\po R_L$ complete the proof for $m=1$.

For $m>1$ the proof is very similar. If $k_1=k_2=\ldots=k_m=1$ it is enough to observe that analogously to \eqref{eq-normal-k=1} we have
\begin{equation*}
\begin{split}
\po w^l_{n_1,n_2,\ldots, n_m+1}(x_1^1,\ldots,x_{n_{m}}^m,x_{n_m+1}^m)=\nu_{p_m,l}(\Lsum_{i=0}^{p_m-1}\po w^l_{n_1,\ldots, n_m}(x_1^1,\ldots,x_{n_m-1}^m, x_{n_m}^m+ i x_{n_m+1}^m)\\
-\Lsum_{i=1}^{p_m-1}\po w^l_{n_1,\ldots,n_i,\ldots, n_m}(x_1^1,\ldots,x_{n_m-1}^m,i \Uplus x_{n_m+1}^m)).
\end{split}
\end{equation*}

Symmetrically we can obtain $\po w^l_{n_1,\ldots,n_j+1,\ldots, n_m}$ for every index $j$. For an induction step on the parameters $k_j$, without loss of generality we also only consider the step $k_m-1 \to k_{m}$. So let us assume the Lemma holds for modules $U$ with group structure $\mathbb{Z}_{p_1^{k_1}}\times\ldots \times\mathbb{Z}_{p_j^{k_m-1}}$. Then we claim that it also holds for $U$ over the group $\mathbb{Z}_{p_1^{k_1}}\times \ldots\times\mathbb{Z}_{p_m^{k_m}}$. 
To prove this claim we can again use the fact that for the term defined by
\begin{equation*}
\begin{split}
\po t^l_{n_1,\ldots,n_m+1}(x_1^1,\ldots,x_{n_m+1}^m)=
\Lsum_{i=0}^{p_m^{k_m}-1}\po w^l_{n_1,\ldots, n_m}(x_1^1,\ldots,x_{n_m-1}^m,x_{n_m}^m\Uplus ix_{n_m+1}^m)\\
\Lplus\Lsum_{i=0}^{p_m^{k_m-1}-1}\po w^l_{n_1,\ldots, n_m}(x_1^1,\ldots,x_{n_m-1}^j,ip_mx_{n_m}^m\Uplus x_{n_m+1}^m).
\end{split}
\end{equation*}
we have $p_m^{k_m-1} \po w^l_{n_1,\ldots, n_m+1} = \po t^l_{n_1,\ldots,n_m+1} + \po r^l_{n_1,\ldots,n_m+1}$,
for some $\po r^l_{n_1,\ldots,n_m+1}$ such that $\po r^l_{n_1,\ldots,n_m+1}(0,0,\ldots,x_{n_m}^m,x_{n_m+1}^m)$ only depends on the values of $x_{n_m}^m$ and $x_{n_m+1}^m$ modulo $p_{m}^{k_m-1}$. The rest of the proof is analogous to the case $m=1$ and we leave it to the reader.
\end{proof}

\section{A recursive principle} 

Let $\m A$ be a finite 2-nilpotent algebra and $\m U, \m L$ be the corresponding modules of coprime order. By Lemma \ref{normal-lemma} we know that then every polynomial of $\m A$ can be written in the form:

\[
\po p ((l_1,u_1),\ldots,(l_n,u_n))
=
\left(\Lsum_{i=1}^n \lambda_i l_i \Lplus \Lsum_{l\in L, c\in U\atop \beta\in U^n}\mu^{l}_{\beta,c} \po w^{l}(\beta \odot \o u \Uplus c),
\Usum_{i=1}^n \alpha_i u_i+u_0 \right).
\]

Some polynomials $\po p$ require $|U|^n$ many $\mu^{l}_{\beta,c} \neq 0$, which might suggest that we cannot efficiently compute this form. In the next section we will however observe that, depending on the input to $\ceqv{A}$, the number of nonzero $\mu^{l}_{\beta,c}$ is polynomial and that $\ceqv{A}$ is essentially equivalent to checking if the expression
\begin{equation}
\label{U-L-expr}
\Lsum_{l\in L, c\in U\atop \beta\in U^n}\mu^{l}_{\beta,c} \po w^{l}(\beta \odot \o x \Uplus c)
\end{equation}
is constant. For now we thus concentrate on analysing properties of expressions (\ref{U-L-expr}). First of all, we would like to simplify (\ref{U-L-expr}) by eliminating constants $c$ and some of $\beta$'s, that are not 'close' to being invertible according to the following definition:

\begin{df}
Let $\beta^{(i)}_1, \beta^{(i)}_2, \ldots, \beta^{(i)}_n \in \mathbb{Z}_{p^k}$. We will say that $\beta^{(i)}$ is nondegenerate if the expression $\Usum_{j=1}^n \beta^{(i)}_j x_j$ can take all values from $\mathbb{Z}_{p^k}$. Moreover $\beta \in U^n$ is nondegenerate, if $\beta^{(i)}$ is nondegenerate for all $1\leq i \leq m$. Moreover, let $(U^n)^*$ denote the set of all nondegenerate $\beta\in U^n$.
\end{df}
Note that $\beta^{(i)}\in (\mathbb{Z}_{p^k})^n$ is nondegenerate iff there is a $j$ such that $\beta^{(i)}_j$ has a multiplicative inverse in $\mathbb{Z}_{p^k}$. Therefore, if in the expression (\ref{U-L-expr}) 
we have some degenerate $\beta \in U^n$, we can find a $d \in U$ such that $\beta = d \cdot\beta'$ (where  $d \cdot \beta'$ is the coordinatewise multiplication) and $\beta'$ is nondegenerate. Thus we can eliminate degenerate expressions and constants by  replacing $\po w^{l}(\beta \odot \o x \Uplus c)$ by $\po w(\beta'\odot \o x)$, where $\po w(u) = \po w^{l}(d\cdot u+c)$. So if $\po m_1, \ldots, \po m_s$ are all the functions from $U\rightarrow L$, then we can transform (\ref{U-L-expr}) to the form 

\begin{equation}
\label{U-L-normalized}
\Lsum_{\substack{\beta \in \m (U^n)^*\\ j =1\ldots s}} \mu_{\beta}^{(j)} \po m_j(\beta \odot \o x).
\end{equation}
Clearly \eqref{U-L-normalized} represents a constant function, if for every fixed $\o x' \in U^n$ it evaluates to the same value, say $c$. We can treat all the equations obtained this way as a system of $|U^n|$ many equations over variables $\mu_{\beta}^{(j)}$:
\[
c = \Lsum_{\substack{\beta \in \m (U^n)^*\\ j =1\ldots s}} \mu_{\beta}^{(j)} \po m_j(\beta \odot \o x')
\]
Solving this system by Gaussian elimination would potentially require exponential time and is thus not the way to go. Our technique will be to sum some of those equation in an organised manner to derive a nice characterisation of the solution set. To do it properly we have to understand how different $\beta$'s interact with each other. For instance there are $\beta, \alpha$ such that values of $\beta \odot \o x$ and $\alpha \odot \o x$ are independent, i.e. for any choice of $c,d \in U$ we can find $\o x$ with $\beta \odot \o x = c$ and $\alpha \odot \o x = d$. On the other hand we can find a pair $\alpha, \beta$ such that value of $\beta \odot \o x$ implies value of $\alpha \odot \o x$, which means that for every $c$ there exist $d$ such that $\beta \odot \o x=c \implies \alpha \odot \o x = d$. We are going to measure the degree of dependence of $\alpha$ and $\beta$ by the concept of $M$-dependence defined below. 

Let $\langle a,b \rangle$ denote the standard inner product (in every module $\Z_{p_i^{k_i}}$). Note that $\beta \odot \overline{x} = \Big(\langle \beta^{(1)}, x^{(1)} \rangle, \ldots, \langle \beta^{(m)}, x^{(m)}\rangle \Big)$. In our definition we will handle these coordinates separately. Now to describe dependencies between nondegenerate $\beta^{(i)}, \alpha^{(i)} \in (\Z_{p_i^{k_i}})^n$ take some invertible $\alpha^{(i)}_j \in \Z_{p_i^{k_i}}$
and put $a = (\alpha^{(i)}_{j})^{-1}\cdot\beta_{j}^{(i)}$. Then $\beta^{(i)} = a \alpha^{(i)} + \widetilde{m}$ for some $\widetilde{m} \in (\Z_{p_i^{k_i}})^n$. It is easy to check, that the image of $f(x^{(i)}) = \langle \widetilde{m}, x^{(i)} \rangle$ does not depend on the choice of $\alpha^{(i)}_j$ and this image is obviously some subgroup of $\Z_{p_i^{k_i}}$. We will call this subgroup $M$. We can see that $\widetilde{m}_j = 0$ so for all $c \in \Z_{p_i^{k_i}}, m \in M$ the system of equation 
$$
\begin{cases}
\langle \alpha^{(i)}, x^{(i)} \rangle = c \\
\langle \widetilde{m}, x^{(i)} \rangle = m
\end{cases}
$$
has a solution (as $\alpha^{(i)}_j$ is invertible and $\widetilde{m}_j = 0$).  We will say that $\beta^{(i)}$ is $M$-dependent on $\alpha^{(i)}$ if $M$ is the image of $\langle \widetilde{m}, x^{(i)} \rangle$. Notice that, if $M=\Z_{p_i^{k_i}}$ then any assumption on the value of $\langle \alpha^{(i)},x^{(i)} \rangle$ does not imply anything on the value of $\langle \beta^{(i)}, x^{(i)} \rangle$, so the expressions $ \langle \beta^{(i)}, x^{(i)} \rangle, \langle \alpha^{(i)}, x^{(i)} \rangle$ are in that sense independent. On the contrary $M=\{ 0 \}$ and $\langle \alpha^{(i)}, x^{(i)} \rangle = c$ implies that $\langle \beta^{(i)},x^{(i)} \rangle = a \cdot c$ for $a = (\alpha^{(i)}_{j})^{-1}\cdot\beta_{j}^{(i)}$ (for $a$ from the definition of dependence).

We give a lemma, summarising some basic properties of $M$-dependence (proof left to the reader). In all cases we assume, that $\beta^{(i)}, \alpha^{(i)} \in (Z_{p_i^{k_i}})^n$ are non-degenerate:

\begin{lm}
\label{lm-util}
\begin{enumerate}
\item The relation of $M$-dependence is symmetric, so if $\beta^{(i)}$ $M$-depends on $\alpha^{(i)}$, then $\alpha^{(i)}$ $M$-depends on $\beta^{(i)}$.
\item Let $\leq(M, \beta^{(i)})$ denote the set of all $\alpha^{(i)} \in \Z_{p_i^{k_i}}$ that are $L$-dependent on $\beta^{(i)}$ for some $L\leq M$. Let $\leq(M)$ be the set of all pairs $(\beta^{(i)},\beta'^{(i)})$ such that $\beta'^{(i)} \in \leq(M, \beta^{(i)})$. Then $\leq(M)$ is an equivalence relation between $\beta^{(i)} \in (\Z_{p_i^{k_i}})^n$.
\item\label{lm-util-1} The following are equivalent
\begin{itemize}
\item $\beta^{(i)} \in (\Z_{p_i^{k_i}})^n$ is $M$-dependent on $\alpha^{(i)}$ 
\item for fixed $x^{(i)}\in \Z_{p_i^{k_i}}$ system of equations 
$$\begin{cases}
\langle \alpha^{(i)}, x^{(i)} \rangle = \langle \alpha^{(i)}, y^{(i)}\rangle\\
\langle \beta^{(i)}, x^{(i)} \rangle = \langle \beta^{(i)}, y^{(i)} \rangle +\ m
\end{cases}
$$
has a solution iff $m \in M$.
\end{itemize}
Moreover the number of solution to the above system of equations for any given $m\in M$ is ${p}^{kn}\over|M|$.
\item\label{lm-util-3} For $\beta \in (U^n)^*$ let $E^{(j)}_\beta (\o x,m)$ be the set of those evaluations $\overline{x}'$, which satisfy the following conditions: $(\beta\odot \o x)^{(j)}+m = (\beta\odot \overline{x}')^{(j)}$ and  $x^{(i)} = x'^{(i)}$ for $i\neq j$. Then $|E^{(j)}_\beta(\o x,m)|=|E^{(j)}_\beta(\o y,m)|$  for every $\o x,\o y\in U^n$.
\end{enumerate}
\end{lm}

For $N\leq \Z_{p_i^{k_i}}$ let $\dep^{(i)}(\beta, N)$ denote the set of all $\alpha$ such that $\alpha^{(i)}$ is $N$-dependent on $\beta^{(i)}$. Let $\leq^{(i)}(N, \beta)$ denote the set of all $\alpha$ with $(\beta^{(i)}, \alpha^{(i)}) \in \leq(N)$. Moreover let $\leq^{(i)}(N)$ be equivalence relation containing all pairs $(\alpha, \beta)$ with $(\alpha^{(i)}, \beta^{(i)}) \in \leq(N)$. Let $M^{(i)}_{k}=\langle p_i^{k}\rangle$ be the subgroup of $\Z_{p_i^{k_i}}$ generated by $p_i^{k}$. 

Notice, that a given expression \eqref{U-L-normalized} represents a constant function if and only if it does not depend on $x^{(i)}$ for any index $i$. The following lemma will therefore be key to construct the recursive algorithm for finite $2$-nilpotent algebra with coprime $|U|$, $|L|$:

\begin{lm}
\label{most-important}
Let $\po m_1, \po m_2, \ldots, \po m_s$ be functions from  $U = \mathbb{Z}_{p_1^{k_1}} \times \mathbb{Z}_{p_2^{k_2}} \times \ldots \times \mathbb{Z}_{p_m^{k_m}}$ to $L$ and let
$$\po t(\o x) = \Lsum_{\substack{\beta \in \m (U^n)^*\\ l =1\ldots s}} \mu_{\beta}^{(l)} \po m_l(\beta \odot \o x).$$
Then, if $t$ does not depend on the variables $x^{(i)}$, we have that for every $\beta \in \m (U^n)^*$:
\begin{equation}\label{eq-most-important}
\Lsum_{\substack{\alpha \in \leq^{(i)}(M^{(i)}_{1}, \beta)\\ l =1\ldots s}} \mu_{\alpha}^{(l)} \po m_l(\alpha \odot \o x) 
\end{equation}
does not depend on variables $x^{(i)}$ for all $i$.
\end{lm}

\begin{proof}
Without loss of generality we assume that $i=1$. Then we put $M_j = M^{(1)}_j$ for all $1\leq j \leq k_1$.
Moreover let $\widetilde{M_j}= \langle (p_1^{j}, 0, \ldots, 0) \rangle$. So $\widetilde{M_j}$ is the subgroup of the underlying group of $\m U$, whereas $M_j$ is the subgroup of $\Z_{p_1^{k_1}}$. We will always write $M_{[0]}$ for $M_0$ (zero in index) to distinguish it from $M_o$ (small letter $o$ in index, which will be used as variable). To prove the theorem we will show by induction, that for $0\le{} j\le{} k_1$ we have that 

\begin{equation}\label{eq-most-importan-ind}
\Lsum_{\substack{\alpha \in \leq^{(1)}(M_1, \beta)\\ l =1\ldots s}} \Lsum_{m \in \widetilde{M_j}} \mu_{\alpha}^{(l)} \po m_l(\alpha \odot \o x+m)
\end{equation}
does not depend on variable $x^{(1)}$. For $j=k_1$, \eqref{eq-most-importan-ind} is exactly the statement of the lemma.

For $j=0$ the expression \eqref{eq-most-importan-ind} obviously does not depend on $x^{(1)}$.  So assume that it does not depend on $x^{(1)}$ for $j'<j$ and we will show that also for $j$ the expression \eqref{eq-most-importan-ind} does not depend on $x^{(1)}$. 

Consider 
\[
\po{Dt}(\o x,\beta,M_j)=\Lsum_{m\in M_j} \Lsum_{\o z\in E^{(1)}_\beta (\o x,m)}\po t(\o z).
\]
By point \ref{lm-util-3} of Lemma \ref{lm-util} and from the fact that $\po t$ does not depend on variables $x^{(1)}$ we have 
\[
\po{Dt}(\o x,\beta,M_j)=\po{Dt}(\o y,\beta,M_j)
\]
for every $\o x,\o y\in U^n$. Pick $\o x$ and $\o y$ such that $x^{(i)}=y^{(i)}$ for $i>1$ and $\beta\in {(U^n)}^*$. By definitions of $\po t$ and $\po{Dt}$ we have that:
\[
\po{Dt}(\o x,\beta, M_{j})=\Lsum_{m\in M_{j}} \Lsum_{\o z\in E^{(1)}_\beta (\o x,m)}\Lsum_{\substack{\alpha \in  (U^n)^*\\ l =1,\ldots, s}} \mu_{\alpha}^{(l)} \po m_l(\alpha \odot \o z)
\]

By the fact that every $\alpha\in (U^n)^*$ is $M_{o}$-dependent on the first coordinate on $\beta$ with exactly one $M_o$ (see definition of dependence) we obtain that:
\[
\po{Dt}(\o x,\beta, M_{j})=\Lsum_{\substack{m\in M_{j}\\l =1,\ldots, s}} \Lsum_{\o z\in E^{(1)}_\beta (\o x,m)}\Lsum_{\substack{\alpha \in \dep^{(1)}( \beta, M_{o})\\ o =0,\ldots, k_1}} \mu_{\alpha}^{(l)} \po m_l(\alpha \odot \o z).
\]
We can regroup summands
\[
\po{Dt}(\o x,\beta,M_{j})=\Lsum_{l =1,\ldots, s} \Lsum_{\substack{\alpha \in \dep^{(1)}( \beta, M_{o})\\ o =0,\ldots, k_1}}\Lsum_{\substack{\o z\in E^{(1)}_\beta (\o x,m)\\m\in M_{j}}} \mu_{\alpha}^{(l)} \po m_l(\alpha \odot \o z).
\]
and then by definition of $E^{(1)}_\beta (\o x,m)$ and by points \ref{lm-util-1} and \ref{lm-util-3} of Lemma \ref{lm-util} we obtain that

\[
\po{Dt}(\o x,\beta, M_{j})=\Lsum_{l =1,\ldots, s} \Lsum_{\substack{\alpha \in \dep^{(1)}( \beta, M_{o})\\ o =0,\ldots, k_1}}\Lsum_{\substack{m'\in \widetilde{M_{o}}\\m\in \widetilde{M_{j}}}} {p_1^{kn}\over|M_{o}|} \mu_{\alpha}^{(l)}  \po m_l(\alpha \odot \o x\Uplus m \Uplus m').
\]
Now, it is easy to see that
\begin{gather*}
\po{Dt}(\o x,\beta,M_{j})=\Lsum_{l =1,\ldots, s} \Lsum_{\substack{\alpha \in \dep^{(1)}( \beta, M_{o})\\ o =0,\ldots, k_1}}\Lsum_{m\in \widetilde{M}_{min\set{o,j}}} { | M_{o}| |M_{j}|\over | M_{min\set{o,j}}| }
{p_1^{kn}\over | M_{o}|} \mu_{\alpha}^{(l)}  \po m_l(\alpha \odot \o x\Uplus m )=\\
=\Lsum_{l =1,\ldots, s} \Lsum_{\substack{\alpha \in \dep^{(1)}( \beta, M_{[o]})\\ o =0,\ldots, k_1}}\Lsum_{m\in \widetilde{M}_{min\set{o,j}}} {p_1^{kn} |M_{j}|\over |M_{min\set{o,j}}| } \mu_{\alpha}^{(l)}  \po m_l(\alpha \odot \o x\Uplus m ).
\end{gather*}
Denote
\[
w_{o,j}={ p_1^{kn} |M_{j}|\over |M_{min\set{o,j}}| }.
\]
After the substitution above we obtain the following:
\[
\po{Dt}(\o x,\beta,M_{j})=\Lsum_{l =1,\ldots, s} \Lsum_{\substack{\alpha \in \dep^{(1)}( \beta, M_{o})\\ o =0,\ldots, k_1}}\Lsum_{m\in \widetilde{M}_{min\set{o,j}}}w_{o,j} \mu_{\alpha}^{(l)}  \po m_l(\alpha \odot \o x\Uplus m ).
\]
Observe that for any $\alpha$  which is $M_{[0]}$-dependent on the first coordinate on $\beta$ and every $l\in\set{1,\ldots,s}$ we have that 
\[
 \Lsum_{m\in \widetilde{M_{[0]}}}w_{0,j} \mu_{\alpha}^{(l)}  \po m_l(\alpha \odot \o x\Uplus m )=\Lsum_{m\in \widetilde{M_{[0]}}}w_{o,j} \mu_{\alpha}^{(l)}  \po m_l(\alpha \odot \o y\Uplus m ).
\]
Hence, for 
\[
\po{Dt'}(\o x,\beta,M_{j})=\Lsum_{l =1,\ldots, s} \Lsum_{\substack{\alpha \in \dep^{(1)}( \beta, M_{o})\\ o =1,\ldots, k_1}}\Lsum_{m\in \widetilde{M}_{min\set{o,j}}}w_{o,j} \mu_{\alpha}^{(l)}  \po m_l(\alpha \odot \o x\Uplus m )
\]
we have that
\[
\po{Dt'}(\o x,\beta, M_{j})=\po{Dt'}(\o y,\beta,M_{j}).
\]
Note that classes of $\leq^{(1)}(M_{j})$ are contained in classes of  $\leq^{(1)}(M_{1})$. Let $\beta_1$, $\beta_2$,\ldots,$\beta_u$ be representants of each $\leq^{(1)}(M_{j})$ class contained in $\leq(M_{1})$ class of $\beta$. As we can substitute $\beta$ with any $\beta_i$ in the previous equation, we get:
\[
\Lsum_{i=1,\ldots,u}\po{Dt'}(\o x,\beta_i,M_{j})=\Lsum_{i=1,\ldots,u}\po{Dt'}(\o y,\beta_i, M_{j}),
\]
and
\[
\Lsum_{i=1,\ldots,u}\po{Dt'}(\o x,\beta_i,M_{j})=\Lsum_{\substack{i=1,\ldots,u	\\l =1,\ldots, s}} \Lsum_{\substack{\alpha \in \dep^{(1)}( \beta_i, M_{o})\\ o =1,\ldots, k_1}}\Lsum_{m\in \widetilde{M}_{min\set{o,j}}}w_{o,j} \mu_{\alpha}^{(l)}  \po m_l(\alpha \odot \o x\Uplus m ).
\]
By rewriting of above expression we obtain that $\Lsum_{i=1,\ldots,u}\po{Dt'}(\o x,\beta_i, M_{j})$ is equal
\[
\Lsum_{\substack{ l=1,\ldots, s\\o =1,\ldots, k_1}} \Lsum_{\substack{\alpha \in \dep^{(1)}( \beta, M_{i})\\ i =1,\ldots, k_1}}\Lsum_{m\in \widetilde{M}_{min\set{o,j}}}\num(\alpha,o)w_{o,j} \mu_{\alpha}^{(l)}  \po m_l(\alpha \odot \o x\Uplus m ),
\]
where $\num(\alpha,o)$ for $\alpha \in \dep^{(1)}( \beta, M_{o})$ is the number of $\beta_i$ from which $\alpha$ is $M_{o}$-dependent on the first coordinate. We have obtained this expression since the $\leq^{(1)}(M_{j})$ classes of $\beta_i$ cover the $\leq^{(1)}(M_{1})$ class of $\beta$ and hence for every $1 \leq o < j$ and $\alpha$ there is $\beta_i$ such that $\alpha$ is $M_{o}$ dependent form $\beta_i$. Moreover, for every such $\alpha$ and $o$ the value of $\num(\alpha,o)$ depends only on $o$. So, in such case we can use $\num(\beta,o)$ instead of $\num(\alpha,o)$.

Note that by induction hypothesis we know that value of the following sum
\[
\Lsum_{\substack{\alpha \in \dep^{(1)}( \beta_i, M_{i})\\ l =1,\ldots, s}} \Lsum_{m \in \widetilde{M_o}} \mu_{\alpha}^{(l)} \po m_l(\alpha \odot \o x+m),
\]
for  $1\leq o < j$  does not depend on first coordinate. Now, by multiplying by constants and adding above sum for $1\leq o < j$  we obtain that 
\[
\Lsum_{\substack{ l=1,\ldots, s\\o =1,\ldots, j-1}} \Lsum_{\substack{\alpha \in \dep^{(1)}( \beta, M_{i})\\ i =1,\ldots, k_1}}\Lsum_{m\in \widetilde{M_{o}}}\num(\beta,o)w_{o,j} \mu_{\alpha}^{(l)}  \po m_l(\alpha \odot \o x\Uplus m ).
\]
does not depend on the  first coordinate. Therefore, we know that 
\begin{gather*}
\Lsum_{\substack{ l=1,\ldots, s\\o =j,\ldots, k_1}} \Lsum_{\substack{\alpha \in \dep^{(1)}( \beta, M_{i})\\ i =1,\ldots, k_1}}\Lsum_{m\in \widetilde{M_{j}}}\num(\alpha,o)w_{o,j} \mu_{\alpha}^{(l)}  \po m_l(\alpha \odot \o x\Uplus m )=\\
=\Lsum_{\substack{ l=1,\ldots, s\\o =j,\ldots, k_1}} \Lsum_{\substack{\alpha \in \dep^{(1)}( \beta, M_{i})\\ i =1,\ldots, k_1}}\Lsum_{m\in \widetilde{M_{j}}}\num(\alpha,o)w_{o,j} \mu_{\alpha}^{(l)}  \po m_l(\alpha \odot \o y\Uplus m ).
\end{gather*}
Notice that for $o\geq j$: $w_{o,j}=p_1^{kn}$. Hence and by the fact that $|U|$ and $|L|$ are coprime we obtain that
\begin{gather*}
\Lsum_{\substack{ l=1,\ldots, s\\o =j,\ldots, k_1}} \Lsum_{\substack{\alpha \in \dep^{(1)}( \beta, M_{i})\\ i =1,\ldots, k_1}}\Lsum_{m\in \widetilde{M_{j}}}\num(\alpha,o) \mu_{\alpha}^{(l)}  \po m_l(\alpha \odot \o x\Uplus m )=\\
=\Lsum_{\substack{ l=1,\ldots, s\\o =j,\ldots, k_1}} \Lsum_{\substack{\alpha \in \dep^{(1)}( \beta, M_{i})\\ i =1,\ldots, k_1}}\Lsum_{m\in \widetilde{M_{j}}}\num(\alpha,o) \mu_{\alpha}^{(l)}  \po m_l(\alpha \odot\o y\Uplus m ).
\end{gather*}
Observe that $\num(\alpha, o)$ for $\alpha\in \dep^{(1)}(\beta,i)$, $i\ge{}1$ is equal $1$ for exactly one $o\geq j$ and $0$ else. Finally, from this facts
\begin{gather*}
\Lsum_{l=1,\ldots, s} \Lsum_{\substack{\alpha \in \dep^{(1)}( \beta, M_{i})\\ i =1,\ldots, k_1}}\Lsum_{m\in \widetilde{M_{j}}} \mu_{\alpha}^{(l)}  \po m_l(\alpha \odot \o x\Uplus m )=\\
=\Lsum_{ l=1,\ldots, s} \Lsum_{\substack{\alpha \in \dep^{(1)}( \beta, M_{i})\\ i =1,\ldots, k_1}}\Lsum_{m\in \widetilde{M_{j}}} \mu_{\alpha}^{(l)}  \po m_l(\alpha \odot \o y\Uplus m ).
\end{gather*}
This completes the proof of \eqref{eq-most-importan-ind} and hence the proof of the lemma.
\end{proof}

\section{Circuit equivalence}
In two previous sections we have  investigated the structure of $2$-nilpotent algebras. We know that every such algebra $\m A$ is of the form $\m L \otimes^{\widehat{F}} \m U$. We have devoted special attention to algebras for which $|L|$ and $|U|$ are co-prime  and we have obtained some useful tools for such algebras. In particular, Lemma \ref{most-important} gives us a method how to reduce our problem to some set of simpler questions.  On the other hand if $\m L$ and $\m U$ are of prime power order for the same prime, then $\m A$ is a supernilpotent algebra  and we can solve \ceqv{} using the algorithm shown by Aichinger and Mudrinski in \cite{aichinger-mudrinski}. Finally, our algorithm, shown in the proof of the next theorem, solves the problem reducing it to cases mentioned above.

\begin{thm}
Let $\m A$ be finite $2$-nilpotent algebra from a congruence modular variety. Then \ceqv{A} is in \ptime. 
\end{thm}

\begin{proof}
Let us consider an input of $\ceqv{A}$, so two circuits that representing two polynomial operations $f,g \in \Pol{\m A}$. Since $\m A$ is a nilpotent algebra from a congruence modular variety it has a Maltsev term $d$ such that for all $a,b \in A$ the function $h(x)=d(x,a,b)$ is a permutation on $A$ (see Lemma 7.3 in \cite{fm}). Hence, to check if $f$ and $g$ describe the same function it is enough to check if the identity $d(f(\overline{x}),g(\overline{x}),0)=0$ holds (for some $0$ of $A$). Therefore, we can assume that our problem is to check if a given circuit with one output gate expresses a function constantly equal $0$. By \cite{fm} we know that there exist modules $\m U$ and $\m L$ such that $\m A=\m L \otimes^F \m U$ and as a consequence all operations over $\m A$ can be expressed in the following form:
\[
\po p^{\m L \otimes^{\widehat{F}}\m U}((l_1,u_1),\ldots,(l_k,u_k))
=
\left(\Lsum_{i=1}^k \lambda_i l_i \Lplus \widehat{\po p}(u_1,\ldots,u_k),
\Usum_{i=1}^k \alpha_i u_i+u_0 \right),
\]
where $\widehat{\po p}$ is a sum of elements
\[
\mu \widehat{f}(\beta\odot\overline{u})
\]
with $f$ being a basic operation of $\m A$, $\beta\in U^n$ and $\mu\in\m R_L$. Note that for algebras with finite signature we can obtain such a form in polynomial time: It is enough to compute these forms in preprocessing for the basic operations of the algebra and then to compute the final form step by step by composing basic operations.

If some $\lambda_i \neq 0$, $\alpha_i \neq 0$ or $u_0 \neq 0$, then $\po p$ obviously does not define a function that is constantly equal to $0$ and so our algorithm will return no. Otherwise it left to check if
\[
\widehat{\po p}(u_1,\ldots,u_k) \equiv 0.
\] 
If the evaluation $\widehat{\po p}(0, \ldots, 0)$ gives a non-zero value then clearly this does not hold. Otherwise, all we need to check is if $\widehat{\po p}(u_1,\ldots,u_k)$ is a constant function.

\begin{claim}{1}\label{claim-coprime}
If $|U|$ and $|L|$ are co-prime, then there exists a polynomial time algorithm to check whether $\widehat{\po p}(u_1,\ldots,u_k)$ is constant.
\end{claim}

We will show that Claim \ref{claim-coprime} holds at the end of the proof. For now assume that it holds. We will then show that it also holds for arbitrary $U$ and $L$. Without loss of generality we can assume that $L=L_1\times\ldots\times L_b$ such that for every $i$ the set $L_i$ is of prime power order and $|L_i|$ and $|L_j|$ are co-prime for $i\not=j$. For checking that $\widehat{\po p}$ is a constant function it is enough to check that its projection on each $L_i$ is constant. Since $\m L$ is affine, it is easy to see that also the algebra decomposes as $\m L=\m L_1\times\ldots\times\m L_b$. 
Hence, the projection of $\widehat{\po p}$ (formally $\widehat{\po p}^{\m L \otimes^{\widehat{F}} \m U}$) on the $i$-th coordinate  is equal to $\widehat{\po p}^{\m L_i \otimes^{\widehat{G}} \m U}$, where $\widehat{G}$ contains projections of operations from $\widehat{F}$ on the $i$-th coordinate. In such a way we can reduce our problem to the case in which $L$ is of prime power order. Let $|L|$ be power of some prime $q$.

Now, again, without loss of generality we can assume that $U=U_1\times U_2$ such that $|U_1|$ is power of $q$ and $q\nmid|U_2|$. 
Using the assumption that $U=U_1\times U_2$ we can divide every argument of $\widehat{\po p}$ into two independent arguments one from $U_1$ and one from $U_2$. Note that if we fix a constant in the arguments from $U_1$ we obtain an operation from $U_2$ to $L$, stemming from a $2$-nilpotent algebra with universe $L \times U_2$. Symmetrically, if we put constants in place of arguments from $U_2$ we obtain a polynomial operation over a $2$-nilpotent algebra with universe $L \times U_1$. Note that the second algebra is a nilpotent algebra of prime power order with finite signature and thus supernilpotent (see \cite{kearnes-spectrum}). \ceqv{} for such algebras can be solved in polynomial time using the algorithm proposed by Aichinger and Mudrinski \cite{aichinger-mudrinski}. In this algorithm to check if a given $n$-ary polynomial operation is constant we need only to check if is constant on a certain set $S$ of $\mathcal O(n^C)$ many tuples, where $C$ is a constant that depends on the algebra.

This enables us to use the following algorithm to solve \ceqv{}. Let $\widehat{\po p}$ be an $n$-ary function. For every evaluation $s$ from $S$, put values from $s$ into arguments from $U_1$ and check if the obtained polynomial over the $2$-nilpotent algebra with universe $L \times U_2$ is constant using the algorithm given by Claim \ref{claim-coprime}. If for every $s\in S$ the obtained polynomial is constant and equal the same constant, then $\widehat{\po p}$ is constant. Otherwise, $\widehat{\po p}$ is not constant. This algorithm obviously solve our problem in polynomial time. Thus all that is left is to give a proof of Claim \ref{claim-coprime}.

If $|U|$ and $|L|$ are co-prime then using the result from the previous two sections we can express $\widehat{\po p}$ as in (\ref{U-L-normalized}):

\begin{equation}\label{eq-U-L-form}
\widehat{\po p}(u_1,\ldots,u_k) =\Lsum_{l=1,\ldots,s\atop \beta\in (U^n)^*}\mu^{l}_{\beta,l} \po m_{l}(\beta \odot \o u) .
\end{equation}

It is not hard to see that we can obtain such a form of $\widehat{\po p}$ in polynomial time step by step composing basic operations occurring in $\po p$. Note that $\widehat{\po p }$ is a constant function if and only if for every $i$ it does not depend on $u^{(i)}$. Hence, from now we will be looking for an algorithm determining if $\widehat{\po p}$ depends on $i$-th coordinate. 

By Lemma \ref{most-important} if our function does not depend on variables $u^{(i)}$ then also for every $\beta \in (U^n)^*$ the following subterm does not depend on $u^{(i)}$:

$$\po t(\overline{u}) = \Lsum_{\substack{j=1\ldots s\\ \alpha \in \leq^{(i)}(M^{(i)}_1, \beta)}} \mu_{\alpha}^{(j)} \po{m}_j(\alpha\odot \overline{u})$$
As $\leq^{(i)}(M^{(i)}_1)$ is an equivalence relation, we can partition the set of nondegenerate $\beta$'s into classes of $\leq^{(i)}(M^{(i)}_1)$ and check if the corresponding expressions $\po t(\overline{u})$ do not depend on $u^{(1)}$. So our algorithm checks if a term $\po t$ for two evaluations $\overline{u},\overline{v}\in U^n$ that differ only on the $i$-th coordinate gives the same value. It considers cases, when $\langle \beta^{(i)}, u^{(i)} \rangle = c$ and $\langle \beta^{(i)}, v^{(i)} \rangle = d$ for all $c,d \in U$. The fact that we sum $\alpha$'s in one $\leq^{(i)}(M^{(i)}_1)$ class will lead us to recursive calls to problems with much simpler terms.

As $\beta^{(i)}$ is nondegenerate, there exist an index $k$ with invertible $\beta^{(i)}_k$. Therefore the equation $\langle \beta^{(i)}, u^{(i)}\rangle  = c$ is equivalent to $u^{(i)}_k =(\beta_k^{(i)})^{-1}(c-\sum_{l\not=k}\beta^{(i)}_l u^{(i)}_l)$. Hence evaluations with $\langle \beta^{(i)}, u^{(i)}\rangle  = c$ are of the form

\begin{equation}
\label{subst}
\Big(u^{(1)}_1,\ldots, u^{(i)}_{k-1}, (\beta_k^{(i)})^{-1}(c-\sum_{l\not=k}\beta^{(i)}_l u^{(i)}_l), u^{(i)}_{k+1}, \ldots, u^{(m)}_n\Big)
\end{equation}

So we can define $\po t_c(\o u)$ as function created from $\po t(\o u)$ by eliminating variable $u^{(i)}_k$ according to \eqref{subst}. Now take $c,d\in U$ and $\overline{u}, \overline{v}$ non-equal only on $i$-th coordinate with $\langle \beta^{(i)}, u^{(i)}\rangle  = c$ and $\langle \beta^{(i)}, v^{(i)}\rangle  = d$.  As $\overline{u},\overline{v}$ are equal on all coordinates different than $i$-th, we identify variables $u^{(l)} = v^{(l)}$ for $l\neq i$ in $\po t_c(\o u)-\po t_c(\o v)$ and get a function $\po w_{c,d}(\o u,\o v)$. Now we recursively check if $\po w_{c,d}(\o u,\o v)\equiv 0$ for all $\o u,\o v$.   It's obvious that the statement that for all choices of $c,d \in U$ equation $\po w_{c,d}(\o u,\o v) \equiv 0$ holds is equivalent to statement that $\po t(\overline{u})$ does not depend on $u^{(i)}$. Now notice, that we can reduce checking $\po w_{c,d}(\o u,\o v)\equiv 0$ to checking if $\po w_{c,d}(\o u,\o v)$ is constant by checking if for one evaluation its 0.

In such a way we will reduce our question to a constant number of easier questions, as there is only a constant number of pairs $c,d$. Note that if $\alpha\in \dep^{(i)}(\beta,M)$ then there exist $\nu_{\alpha,\beta,i}\in \Z_{p_i^{k^i}}$ and $m_{\alpha,\beta,i}$ such that  $\langle \alpha^{(i)}, u^{(i)}\rangle  = \nu_{\alpha,\beta,i} \langle \beta^{(i)}, u^{(i)}\rangle  + \langle m_{\alpha,\beta,i},u^{(i)}\rangle $ and $\langle m_{\alpha,\beta,i}, u^{(i)}\rangle \in M$. Hence, we obtain that 

\[
\po t(\overline{u}) = \Lsum_{\substack{j=1\ldots s\\ \alpha \in \leq^{(i)}(M^{(i)}_1, \beta)}} \mu_{\alpha}^{(j)} \po{m}_j(\langle \alpha^{(1)}, u^{(1)}\rangle , \ldots, \nu_{\alpha,\beta,i} \langle \beta^{(i)}, u^{(i)}\rangle  + \langle m_{\alpha,\beta,i}, u^{(i)}\rangle , \ldots, \langle \alpha^{(m)}, u^{(m)}\rangle )
\] 
and if we assume that $\langle \beta^{(i)}, u^{(i)}\rangle  = c$ and consequently substitute $u^{(i)}_k$ according to \eqref{subst} then 

\begin{equation}\label{eq-alg}
 \po t_c (\o u)= \Lsum_{\substack{j=1\ldots s\\ \alpha \in \leq^{(i)}(M^{(i)}_1, \beta)}} \mu_{\alpha}^{(j)} \po{m}_j(\langle \alpha^{(1)}, u^{(1)}\rangle , \ldots, \nu_{\alpha,\beta,i} c + \langle m'_{\alpha,\beta,i}, u^{(i)}\rangle , \ldots, \langle \alpha^{(m)}, u^{(m)}\rangle )
\end{equation}
where  $\langle m'_{\alpha,\beta,i}, u^{(i)}\rangle \in M^{(i)}_1$.  Now observe, that since $\langle m'_{\alpha,\beta,i}, u^{(i)}\rangle \in M^{(i)}_1$ we have that $\langle m'_{\alpha,\beta,i}, u^{(i)}\rangle =p_i\cdot(\langle m''_{\alpha,\beta,i}, u^{(i)}\rangle )=\langle m''_{\alpha,\beta,i},(p_i\cdot u^{(i)})\rangle $. Since $p_i\Z_{p_i^{k_i}}$ is isomorphic to $Z_{p_i^{k_i -1}}$, we obtain, in fact, that the expression $\po t_c(\o u)-\po t_d(\o v)$ is over some smaller domain (as we can apply reasoning to both $\po t(\o u)$ and $\po t(\o v)$). So $\po w_{c,d}(\o u,\o v)$ can be regarded as an expression over smaller domain. We then can continue recursively and apply Lemma \ref{most-important} to $\po w_{c,d}(\o u,\o v)$ ( with the remark that $\po w_{c,d}(\o u,\o v)$ might not be in the form required, but can be easily turned into such an expression, by eliminating all degenerated expressions as discussed in the last section).

This consideration gives us a recursive algorithm for determining if a given function in form \eqref{eq-U-L-form} is constant. Observe that if for all $i$ we have $k_i=0$, then $U$ is one element set so we obtain constant expressions and all we need is to compare their values.  Note that in every recursive call we reduce the problem to solving linearly many simpler cases. Simpler means for us that this new functions have descriptions which are not longer than twice the original one and are over smaller domain. Since the depth of the recursion is bounded by a constant it means that our algorithm works in polynomial time. This observation completes the proof of the claim and in a consequence the proof of the theorem.
\end{proof}

\end{document}